\documentclass[letterpaper]{article} 
\usepackage{aaai25}  
\nocopyright
\usepackage{times}  
\usepackage{helvet}  
\usepackage{courier}  
\usepackage[hyphens]{url}  
\usepackage{graphicx} 
\usepackage{natbib}  
\usepackage{caption} 
\usepackage{algorithm}
\usepackage{algorithmic}
\usepackage{subfigure}
\usepackage{amsmath}
\usepackage{amssymb}
\usepackage{mathtools}
\usepackage{amsthm}
\newtheorem{theorem}{Theorem}[section]
\newtheorem{proposition}[theorem]{Proposition}
\newtheorem{lemma}[theorem]{Lemma}
\newtheorem{corollary}[theorem]{Corollary}
\theoremstyle{definition}
\newtheorem{definition}[theorem]{Definition}
\newtheorem{assumption}[theorem]{Assumption}
\theoremstyle{remark}

\DeclareMathOperator*{\argmin}{arg\,min}

\usepackage{newfloat}
\usepackage{listings}
\DeclareCaptionStyle{ruled}{labelfont=normalfont,labelsep=colon,strut=off} 
\floatstyle{ruled}
\newfloat{listing}{tb}{lst}{}
\floatname{listing}{Listing}
\title{Approximate Bilevel Difference Convex Programming for \\ Bayesian Risk Markov Decision Processes}
\author {
    Yifan Lin\textsuperscript{\rm 1},
    Enlu Zhou\textsuperscript{\rm 1}
}
\affiliations {
    \textsuperscript{\rm 1}
    Industrial and Systems Engineering\\ 
    Georgia Institute of Technology\\ 
    755 Ferst Dr\\
    Atlanta, GA 30332 USA\\
    ylin429@gatech.edu, enlu.zhou@isye.gatech.edu
}

\begin{document}

\maketitle

\begin{abstract}
We consider infinite-horizon Markov Decision Processes where parameters, such as transition probabilities, are unknown and estimated from data. The popular distributionally robust approach to addressing the parameter uncertainty can sometimes be overly conservative. In this paper, we utilize the recently proposed formulation, Bayesian risk Markov Decision Process (BR-MDP), to address parameter (or epistemic) uncertainty in MDPs. To solve the infinite-horizon BR-MDP with a class of convex risk measures, we propose a computationally efficient approach called approximate bilevel difference convex programming (ABDCP). The optimization is performed offline and produces the optimal policy that is represented as a finite state controller with desirable performance guarantees. We also demonstrate the empirical performance of the BR-MDP formulation and the proposed algorithm. 
\end{abstract}

\section{Introduction}
In a Markov decision process (MDP), an agent must make decisions in a sequence while facing uncertainty. In this situation, some parameters of the MDP, such as the transition probabilities and costs, may be unknown and must be estimated from available data. The problem then becomes how to determine the best course of action, given the limited or possibly absent data, in order to minimize the expected total cost and optimize the decision-making process under these uncertain parameters.

An alternative approach to addressing the epistemic uncertainty in MDP is through the use of distributionally robust MDPs (DR-MDP, \citet{xu2010distributionally}). It considers unknown parameters as random variables and assumes that their distributions belong to an ambiguity set determined by the available data. The optimal policy is then found by minimizing the expected total cost using the most adversarial distribution within this ambiguity set. However, these distributionally robust approaches may lead to overly conservative solutions that do not perform well in scenarios that are more likely to occur than the worst case. Additionally, the DR-MDP framework does not explicitly incorporate the dynamics of the problem, as the distribution of the unknown parameters does not depend on the data process, and is therefore not time-consistent, as noted by \citet{shapiro2021tutorial}. In light of these limitations, \citet{lin2022bayesian} propose a Bayesian risk MDP (BR-MDP) framework to address epistemic uncertainty in MDPs. However, the approximation algorithm proposed by \citet{lin2022bayesian} only applies to finite-horizon MDPs and does not scale well with long horizon. It only provides an upper bound on the exact value function, without any theoretical guarantee on the gap.

We reformulate the considered BR-MDP as a bilevel difference convex program (DCP) such that we can employ the powerful optimization methods for DCP to solve infinite-horizon BR-MDP. Since the space of posterior distributions (beliefs) is uncountably infinite, we approximate the bilevel DCP by considering only a subset of posterior distributions. Although the DCP is approximate, we show that its solution is a lower bound on the exact optimal value function. Using the representation of a finite state controller of the resulting policy, we further show an upper bound on the exact optimal value function. We develop an iterative approach to reduce the gap between upper and lower bounds by incrementally generating new sets of posterior distributions, and show the convergence of the proposed algorithm. 

To summarize, the contributions of this paper are two-fold. First, we analyze the infinite-horizon MDP with epistemic uncertainty under BR-MDP via a Bayesian perspective and show the existence and uniqueness of stationary optimal policy. Second, we propose an approximate difference convex programming algorithm to solve the proposed formulation and show the convergence of the proposed algorithm. The rest of the paper is organized as follows. We conduct literature review and introduce the BR-MDP framework in Section \ref{sec: literature}. We show the existence and uniqueness of a stationary optimal policy to the infinite-horizon BR-MDP in Section \ref{sec: Bellman}. We provide a bilevel DCP solution to the infinite-horizon BR-MDP in Section \ref{sec: DCP}. A computationally efficient approximate DCP algorithm is shown in Section \ref{sec: approximate DCP}. We verify the theoretical results and demonstrate the performance of our algorithms via numerical experiments in Section \ref{sec: experiment}. Finally, we conclude the paper in Section \ref{sec: conclusions}.

\section{Background}\label{sec: literature}
\subsection{Related Literature}
If data used to estimate the true but unknown underlying MDP are not sufficient, the estimated MDP may significantly differ from the true MDP, leading to poor policy performance. This discrepancy (between the estimated MDP and the true MDP) can be seen tightly linked to the epistemic uncertainty about the model. There have been numerous approaches that address epistemic uncertainty in MDPs, with robust MDP and its variants (\citet{nilim2003robustness, iyengar2005robust, delage2010percentile, wiesemann2013robust, petrik2019beyond, zhou2021finite, yang2022toward, cousins2024percentile, derman2020bayesian}) being one of the most widely used methods. In robust MDPs, the optimal decisions are made based on their performance under the most unfavorable conditions within a known ambiguity set of possible parameter values.

Apart from the overly conservative robust MDP approach which only considers the worst-case scenario, the risk-averse approach has been proposed to address the epistemic uncertainty, but with more flexibility in choosing the risk functional. Risk-averse approach is originally proposed to address the aleatoric uncertainty due to the inherent stochasticity of the underlying MDP (\citet{howard1972risk, ruszczynski2010risk, petrik2012approximate, osogami2012robustness}). It replaces the risk-neutral expectation by some general risk measures, such as conditional value-at-risk (CVaR, see \citet{rockafellar2000optimization}). However, most of the existing approaches assume the agent has access to the true underlying MDP, and optimize some risk measures such as CVaR in that single MDP (\citet{chow2014algorithms, tamar2015policy, tamar2015optimizing, sharma2019robust}). In this paper, we consider the offline planning problem in MDPs, where we only have access to prior belief distribution over MDPs that is constructed by the offline data. It should be noted that offline planning problem has also been considered by \citet{duff2002optimal}, where the author proposes a Bayes-adaptive MDP (BA-MDP) formulation with an augmented state composed of the underlying MDP state and the posterior distribution of the unknown parameters. Mostly close to the problem setting in this work are \citet{rigter2021risk, lin2022bayesian}. \citet{rigter2021risk} optimize a CVaR risk functional over the total cost and simultaneously addresses both epistemic and aleatoric uncertainty, while \citet{lin2022bayesian} consider a nested risk functional to ensure the time consistency of the obtained policy. 

While there are many works proposing different models and frameworks to address the epistemic uncertainty, developing computationally efficient solutions is also of great interest. In robust MDPs, with some mild conditions on the ambiguity set such as rectangularity, the proposed formulation can be solved by a second-order cone program when the horizon is finite, or policy iteration when the horizon is infinite (\citet{mannor2019data}). In BA-MDP and its variants, \citet{rigter2021risk} propose an approximate algorithm based on Monte Carlo tree search and Bayesian optimization. \citet{lin2022bayesian} develop an $\alpha$-function approximation algorithm using the convexity of the CVaR risk measure. However, the aforementioned works consider a finite-horizon MDP and do not generalize well to the infinite-horizon setting. 

Compared to standard MDPs, our considered problem has two distinct features that make it difficult to apply value iteration, policy iteration, or linear programming (\citet{puterman2014markov}). First is the resulting continuous-state MDP due to augmented belief state. We note that this continuous-state MDP is similar to a belief-MDP, which is the equivalent way to represent a partially observable MDP (POMDP) by treating the posterior distribution of the hidden state as a belief state. Second is the risk measure taken with respect to the unknown parameters in the MDPs. In this work, we propose an optimization-based method to solve the infinite-horizon BR-MDPs. It has been empirically shown by \citet{alagoz2015optimally} that linear programming can efficiently solve a significant number of MDPs in comparison to standard dynamic programming methods, such as value iteration and policy iteration. Furthermore, linear programming requires less memory and can handle MDPs with a larger number of states and still achieve optimality. It has been widely used in risk-sensitive MDP (that deals with intrinsic or aleatoric uncertainty that is due to the inherent stochasticity of the underlying MDP, see \citet{zhang2021cautious}). Works that are most related to our proposed optimization-based approach include \citet{poupart2015approximate} who propose an approximate linear programming algorithm for the risk-neutral constrained POMDPs, and \citet{ahmadi2021constrained} who propose a difference convex program (DCP) for the constrained risk-averse MDPs. Our approach for infinite-horizon BR-MDP significantly differs from the above approaches in two aspects. First, compared to the linear programming approach in \citet{poupart2015approximate}, we use bilevel DCP, due to the additional risk measure used for mitigating the epistemic uncertainty. Our considered risk measure brings additional challenge to exactly evaluating the policy, whereas policy evaluation can be easily solved by a system of linear equations in \citet{poupart2015approximate}. Second, compared to the DCP for the risk-averse MDP with aleatoric uncertainty in \citet{ahmadi2021constrained}, the resulting continuous-state MDP in our problem has an infinite number of constraints and requires appropriate approximation to make the problem computationally feasible.

\subsection{Preliminary: Bayesian Risk MDPs}
Consider an infinite-horizon MDP that is defined as $(\mathcal{S}, \mathcal{A}, P, C, \gamma)$, where $\mathcal{S}$ is the state space, $\mathcal{A}$ is the action space, $P$ is the transition probability with $P(s'|s,a)$ denoting the probability of transitioning to state $s'$ from state $s$ when action $a$ is taken, $C$ is the cost function with $C(s,a,s')$ denoting the cost when action $a$ is taken and state transitions from $s$ to $s'$, $0 \leq \gamma < 1$ is the discount factor. We assume the state space and action space are finite and the cost is bounded. A Markovian deterministic policy $\pi$ is a function mapping from $\mathcal{S}$ to $\mathcal{A}$. Given an initial state $s$, the goal is to find an optimal policy that minimizes the expected discounted total cost: $\underset{\pi}{\min} \mathbb{E}^{\pi, P, C}\left[\sum_{t=1}^{\infty} \gamma^{t-1} C\left(s_{t}, a_{t},s_{t+1}\right)|s_1=s\right]$, where $\mathbb{E}^{\pi, P, C}$ is the expectation with policy $\pi$ when the transition probability is $P$ and the cost is $C$. In practice, $P$ and $C$ are often unknown and estimated from data.

BR-MDP is a recently proposed framework that deals with the epistemic uncertainty in MDPs (see \citet{lin2022bayesian}). It is assumed that the state transition is specified by the state equation $s'=g(s,a,\xi)$ with a known transition function $g$ which involves state $s \in \mathcal{S}\subseteq \mathbb{R}^{k_s}$, action $a \in \mathcal{A}\subseteq \mathbb{R}^{k_a}$, and randomness $\xi \in \Xi \subseteq \mathbb{R}^{k_\xi}$, where $k_s$, $k_a$, $k_{\xi}$ are the dimensions of the state, action, and randomness space, respectively. The state equation together with the distribution of $\xi$ uniquely determines the transition probability of the MDP, i.e., $P(s'\in S'|s,a)=P(\{\xi \in \Xi: g(s,a,\xi)\in S'\}|s,a)$, where $S'$ is a measurable set in $\mathcal{S}$. We refer the readers to Chapter 3.5 in \citet{puterman2014markov} for the equivalence between stochastic optimal control and MDP formulation. We use the representation of state equations instead of transition probabilities in MDPs, for the purpose of decoupling the randomness and the policy, leading to a cleaner formulation in the nested form. The cost is assumed to be a function of state $s$, action $a$, and randomness $\xi$, i.e., $C(s,a,\xi)$. 

The distribution of $\xi$, denoted by $f(\cdot;\theta^c)$, is assumed to belong to a parametric family $\left\{f(\cdot;\theta)| \theta \in \Theta \right\}$, where $\Theta \subseteq \mathbb{R}^{d}$ is a convex parameter space, $d$ is the dimension of the parameter space $\Theta$, and $\theta^c \in \Theta$ is the true but unknown parameter value. Many real-world problems exhibit the characteristic of relying on a parametric assumption. For example, it is commonly assumed that the demand of customers follows a Poisson distribution with an unknown arrival rate in inventory control. We begin by assuming a prior distribution, denoted by $\mu$, over the parameter space $\Theta$. This prior accounts for the uncertainty of the parameter estimate that comes from an initial set of data, and it can also take expert opinions into consideration. Then, given an observed realization of the data process, we update the posterior distribution $\mu$ according to the Bayes' rule. Let the policy be a sequence of mappings from state $s$ and posterior $\mu$ to the action space, i.e., $\pi = \{\pi: \mathcal{S} \times \mathcal{M} \to \mathcal{A}\}$, where $\mathcal{M}$ is the space of posterior distributions. This representation implies the policy is stationary. Now we present the BR-MDP formulation below. 
\begin{align}\label{formulation: BR-MDP}
    & \underset{\pi}{\min} \hspace{0.2cm} \rho_{\mu_1} \mathbb{E}_{\theta_1} \Big[C_1 (s_1,a_1,\xi_1)  + \cdots  + \gamma^{t-1} \rho_{\mu_{t}} \mathbb{E}_{\theta_{t}} \big[ \nonumber \\ 
    & ~~~~~~ C_t(s_{t}, a_{t}, \xi_{t}) + \cdots \big] | s_1=s, \mu_1=\mu \Big]  \\
    & \text{s.t.} \hspace{0.2cm} ~~~~~s_{t+1} = g(s_t,a_t,\xi_t), \label{state_eqn} \\
    & ~~~~~~~~~~~~~~~~ a_{t} = \pi(s_t,\mu_t), \nonumber \\
    &~~~~~~~\mu_{t+1}(\theta) = \frac{\mu_{t}(\theta) f\left(\xi_{t};\theta \right)}{\int_{\Theta} \mu_{t}(\theta) f\left(\xi_{t};\theta \right)d\theta},
    \label{posterior_update} 
\end{align}
where $\rho$ is a risk measure (we defer the definition and form of the risk measure $\rho$ to Section~\ref{sec: risk measure}), $\theta_t$ is a random vector following distribution $\mu_t$, $\mathbb{E}_{\theta_t}$ denotes the expectation with respect to $\xi_t \sim f(\cdot;\theta_t)$ conditional on $\theta_t$, and $\rho_{\mu_t}$ denotes a risk functional with respect to $\theta_t \sim \mu_t$ applied in nested form to the expected total cost with respect to the Bayesian posterior distributions of the unknown parameters. Equation~\eqref{state_eqn} is the transition of the state $s_t$, and without loss of generality we assume the initial state $s_1$ takes a deterministic value $s$. Equation~\eqref{posterior_update} is the updating of the posterior $\mu_t$. 
For a given dataset with size $N$, the prior distribution converges to a Dirac delta function concentrated on the true parameter $\theta^c$ with probability 1, and the optimal value function of BR-MDP converges to the optimal value function of the true MDP. 
 
\subsection{Preliminary: Risk Measure}\label{sec: risk measure}
Let $(\Omega, \mathcal{F}, \mathbb{P})$ be a probability space and $\mathcal{Z}$ be a linear space of $\mathcal{F}$-measurable functions $Z: \Omega \rightarrow \mathbb{R}$. A risk measure is a function $\rho: \mathcal{Z} \to \mathbb{R}$ which assigns a random variable $Z$ to a real number representing its risk. It is said that risk measure $\rho$ is convex if it possesses the properties of convexity, monotonicity, and translation invariance (see \citet{follmer2002convex}). In this paper we consider a class of convex risk measures which can be represented in the following parametric form: $\rho_{\mu}(Z):=\inf_{\phi \in \Phi} \mathbb{E}_{\mu}[\Psi(Z, \phi)]$, where $\Phi \subset \mathbb{R}^{m}$ and $\Psi: \mathbb{R} \times \Phi \to \mathbb{R}$ is a real-valued convex function, and $\Psi(\cdot,\phi)$ is finite-valued and continuous on a compact set of $\phi$. There is a large class of risk measures which can be represented in the parametric form. For example, conditional value-at-risk (CVaR), defined as $\operatorname{CVaR}_{\alpha}(X)=\min_{\phi \in \mathbb{R}} \left\{\phi+\frac{1}{1-\alpha} \mathbb{E}\left[(X-\phi)^{+}\right]\right\}$, where $(\cdot)^{+}$ stands for $\max (0,\cdot)$, is widely used (see \citet{rigter2021risk, chow2015risk}). Another example is risk measures constructed from $\phi$-divergence ambiguity sets (see Example 3 in \citet{guigues2021risk}). We refer the readers to \citet{shapiro2021lectures} for a comprehensive discussion. 

\section{Algorithm and Analysis}
\subsection{Bellman Equation and Optimality}\label{sec: Bellman}
We can write the value function under policy $\pi$ of BR-MDP in the following recursive forms. 
\begin{align*}
    & V^{\pi}(s,\mu) = \rho_{\mu}\mathbb{E}_{\theta}\big[C(s,a,\xi) + \gamma V^{\pi}(s',\mu')\big] \\
    & \text{s.t.} \hspace{0.2cm} ~~~~s'=g(s,a,\xi), a=\pi(s,\mu); \\
    & ~~~~~ \mu'(\theta)=\frac{\mu(\theta) f\left(\xi;\theta \right)}{\int_{\Theta} \mu(\theta) f\left(\xi;\theta \right)d\theta}.
\end{align*}
We refer the readers to \citet{lin2022bayesian} for a discussion on the preference of dynamic risk measure over static risk measure in consideration of time consistency and derivation of the Bellman equation. For simplicity we only consider deterministic policies, but all the analysis below can be extended to stochastic policies. For the stochastic policies, the expectation in \eqref{formulation: BR-MDP} is taken with respect to the randomness $\xi$ and the action $a$. As a consequence of Theorem 5.5.3b in \citet{puterman2014markov}, it is sufficient to consider the Markovian policy. The optimal value function is then denoted as $V^{*}(s,\mu)=\min_{\pi \in \Pi^{MD}} V^{\pi}(s,\mu)$, where $\Pi^{MD}$ is the set of Markovian deterministic policies. It should be noted that the Bayes optimality is with respect to the prior belief $\mu$. In the following, we derive the intermediate results to show $V^{*}$ is the unique optimal value function to the infinite-horizon BR-MDP.

\begin{definition}[Bellman Operator]\label{Bellman operator}
Let $B(\mathcal{S},\mathcal{M})$ be the space of real-valued bounded measurable functions on $(\mathcal{S} \times \mathcal{M})$. For any bounded value function $V \in B(\mathcal{S}, \mathcal{M})$, define an operator $\mathcal{T}: B(s,\mu) \to B(s,\mu)$ as:
\begin{align*}
    (\mathcal{T} V) (s,\mu) = \min_{a \in \mathcal{A}} \rho_{\mu}\left[\mathbb{E}_{\theta}\left[C(s,a,\xi)+\gamma V(s',\mu')\right]\right].
\end{align*}
Also let $\mathcal{T}^{\pi}: B(s,\mu) \to B(s,\mu)$, where
\begin{align*}
    (\mathcal{T}^{\pi} V) (s,\mu) = \rho_{\mu}\left[\mathbb{E}_{\theta}\left[C(s,\pi(s,\mu),\xi)+\gamma V^{\pi}(s',\mu')\right]\right].
\end{align*}
\end{definition}

The next two lemmas show the above Bellman operators are monotonic and contraction mappings.

\begin{lemma}[Monotonicity]\label{Monotonicity}
The operators $\mathcal{T}^{\pi}$ and $\mathcal{T}$ are monotonic, in the sense that $V \leq V'$ implies $\mathcal{T}^{\pi} V \leq \mathcal{T}^{\pi} V'$ and $\mathcal{T} V \leq \mathcal{T} V'$.
\end{lemma}

\begin{lemma}[Contraction Mapping]\label{Contraction mapping}
The operators $\mathcal{T}^{\pi}$ and $\mathcal{T}$  are $\gamma$ contraction for $||\cdot||_{\infty}$ norm. That is, for any two bounded value functions $V, V' \in B(\mathcal{S}, \mathcal{M})$, we have 
\begin{align*}
    ||\mathcal{T}^{\pi} V - \mathcal{T}^{\pi} V'||_{\infty} \leq \gamma ||V-V'||_{\infty}.
\end{align*}
\end{lemma}

The following proposition shows that sub-solutions $V_{\operatorname{sub}}$ and super-solutions $V_{\operatorname{sup}}$ of the optimality equations $V = \mathcal{T}V$ provide lower and upper bounds on $V^{*}$. As a result, when a solution is obtained, both bounds are satisfied, meaning that the solution must be equivalent to $V^{*}$. Additionally, this outcome serves as an important algorithmic tool for optimization-based methods.

\begin{proposition}\label{LP pre-condition}
For any $V_{\operatorname{sub}}, V_{\operatorname{sup}} \in B(\mathcal{S}, \mathcal{M})$, (i) if $V_{\operatorname{sup}} \geq \mathcal{T}V_{\operatorname{sup}}$, then $V_{\operatorname{sup}} \geq V^{*}$; (ii) if $V_{\operatorname{sub}} \leq \mathcal{T}V_{\operatorname{sub}}$, then $V_{\operatorname{sub}} \leq V^{*}$.
\end{proposition}

According to Proposition \ref{LP pre-condition}, we have $V^{*}=\mathcal{T}V^{*}$. By Banach fixed-point theorem, $V^{*}$ is the unique optimal value function to the infinite horizon BR-MDP. We also have that the value $V$ of a stationary policy $\pi$ is the unique bounded solution of the equation $V=\mathcal{T}^{\pi}V$. Similar analysis shows the existence and uniqueness of the optimal stationary policy $\pi^{*}$ that satisfies $V^{*}=\mathcal{T}^{\pi^{*}}V^{*}$. 

Applying the operator $\mathcal{T}$ on any initial value function $V$, we have the value iteration algorithm for the infinite-horizon BR-MDP problem. The following corollary of convergence rate is similar to the standard with the contraction property. 

\begin{corollary}
For any initial bounded value function $V$, the convergence rate is shown to be $||(\mathcal{T}^{k}V)(s,\mu) - V^{*}(s,\mu)||_{\infty} \leq \gamma^{k} ||V(s,\mu)-V^{*}(s,\mu)||_{\infty}$.
\end{corollary}

\subsection{Bilevel Difference Convex Programming}\label{sec: DCP}
The main challenge of executing the value iteration algorithm (and similarly policy iteration algorithm) lies in the continuous augmented state. In this work, we propose an optimization-based method to solve the infinite-horizon BR-MDPs. According to Proposition \ref{LP pre-condition}, the infinite-horizon BR-MDP can be solved as follows:
\begin{align*}
    & \max_V \sum_{s \in \mathcal{S},\mu \in \mathcal{M}} \alpha(s,\mu) V(s,\mu) \\
    & ~~ \text{s.t. } V(s,\mu) \leq \rho_{\mu} \mathbb{E}_{\theta} \Big[C(s,a,\xi) + \gamma V(s',\mu')\Big], \\
    & ~~~~~~~~\forall a \in \mathcal{A}, s \in \mathcal{S}, \mu \in \mathcal{M},
\end{align*}
where we choose $\alpha(s,\mu)$ to be positive scalars which satisfy $\sum_{s \in \mathcal{S}, \mu \in \mathcal{M}} \alpha(s,\mu) = 1$. For the considered class of convex risk measures, we can rewrite the above formulation as a bilevel difference convex program:
\begin{align}
    & \min_V \quad -\sum_{s \in \mathcal{S},\mu \in \mathcal{M}} \alpha(s,\mu) V(s,\mu) \label{DCP} \\
    & \text{s.t.} V(s,\mu) - \min_{\phi} \mathbb{E}_{\mu}\Big[\Psi\big(\mathbb{E}_{\theta}[C(s,a,\xi) + \gamma V(s',\mu')], \phi\big)\Big] \nonumber \\
    & ~~~~ \leq 0, \forall a \in \mathcal{A}, s \in \mathcal{S}, \mu \in \mathcal{M} \nonumber.
\end{align}

Since $\Psi(Z,\phi)$ is convex in $(Z,\phi)$ and expectation is a linear operator, the minimum of $\mathbb{E}[\Psi(Z,\phi)]$ over a convex set $\Theta$ remains convex in $Z$. Thus, \eqref{DCP} is a bilevel difference convex program (see \citet{horst1999dc} for the definition of DCP). It should be noted that \citet{ahmadi2021constrained} show that the minimum over $\phi$ can be absorbed into the overall minimum problem, and $\phi$ is treated as a single variable. However, it is clear that the minimum is achieved at different $\phi$ for different augmented state $(s,\mu)$, thus turning \eqref{DCP} into a bilevel optimization problem. When the lower-level problem is convex and satisfies certain regularity conditions, we can use the Karush-Kuhn-Tucker (KKT) conditions to reformulate the lower-level optimization problem, which allows us to transform the original bilevel optimization problem into a single-level (constrained) optimization problem. 

After being reduced to a single-level DCP problem, \eqref{DCP} can be solved by the convex-concave procedure (see \citet{lipp2016variations} for such procedure), wherein the concave terms are replaced by a convex upper bound. We employ the method of disciplined convex-concave programming (DCCP, \citet{shen2016disciplined}), which converts a DCP problem into a disciplined convex program and subsequently into an equivalent cone program. However, one problem remains to be solved: the number of constraints in \eqref{DCP} is infinite, due to the continuous belief state. To tackle this problem, we take a similar approach as \citet{poupart2015approximate}. The main idea is to start with a finite posterior set (belief space) $\hat{\mathcal{M}}$, and then problem \eqref{DCP} can be solved efficiently by DCCP, where the posterior distribution (belief point) not in the set $\hat{\mathcal{M}}$ is replaced by convex combination of the points in $\hat{\mathcal{M}}$. We then iteratively add to the posterior set new posterior distributions that are reachable from the current set and re-solve \eqref{DCP}. It should be noted that the proposed approach could be extended to value iteration and policy iteration, but the analysis would be more complicated, since there would be a trade-off between the optimization (e.g. value iteration) and the belief point generation, and it could be quite tricky to decide the optimal number of steps for value iteration and optimal intervals for belief point generation. We formally introduce the approximate bilevel DCP algorithm in the next section. 

\subsection{Approximate Bilevel Difference Convex Programming}\label{sec: approximate DCP}
Let $\hat{\mathcal{M}}$ be the current posterior set. Let $\mu^{sas'}$ be the one-step posterior distribution with observed randomness $\xi$ indicated by state transition $s'=g(s,a,\xi)$ and current posterior $\mu$. Initially the posterior set is constructed from corner (degenerate) points. In case the parameter space $\Theta$ is finite, the corner points are $(1,0,\cdots,0)$, $(0,1,0,\cdots)$, $\cdots,$ and $(0,\cdots,0,1)$. In case the parameter space is continuous, it is impossible to express one-step posterior distribution (i.e., $\mu^{sas'}$) as a convex combination of those degenerate points. Therefore, we assume the parameter space is finite, which is practical in many real-world problems. It can also be viewed as a discrete approximation of a continuous parameter set, and the discretization can be chosen of any precision.

\begin{algorithm}[htp]
\caption{Approximate Bilevel DCP}
\label{alg: approximate DCP}
\begin{algorithmic}
\STATE {\bfseries input:} posterior set $\hat{\mathcal{M}}$
\STATE {\bfseries output:} policy $\hat{\pi}^{*}$, approximate value function $\hat{V}^{*}$
\STATE 1. solve the following approximate bilevel DCP:
{\small
\begin{align}\label{ADCP}
    & \min_V \quad - \sum_{s \in \mathcal{S},\mu \in \hat{\mathcal{M}}} \alpha(s,\mu) V(s,\mu) \\
    & \text{s.t. } V(s,\mu) \leq \min_{\phi} \sum_{\theta \in \Theta} \mu(\theta)\Big[\Psi\big(\gamma \sum_{\mu' \in \hat{\mathcal{M}}, s' \in \mathcal{S}} P(s'|s,a,\theta) \nonumber \\
    & w(\mu', \mu^{sas'}) V(s',\mu') + C(s,a,\theta), \phi\big)\Big], \forall a \in \mathcal{A}, s \in \mathcal{S}, \mu \in \hat{\mathcal{M}} \nonumber
\end{align}
}where $w(\mu',\mu^{sas'})$ is obtained by solving \eqref{Weight LP}.
\STATE 2. obtain the approximate solution $\hat{V}^{*}$ to \eqref{ADCP}; obtain the approximate policy
{\small
\begin{align*}
    & \hat{\pi}^{*}(s,\mu)=\argmin_{\phi, a \in \mathcal{A}} \sum_{\theta \in \Theta} \mu(\theta)\Big[\Psi\big(\gamma \sum_{\mu' \in \hat{\mathcal{M}}, s' \in \mathcal{S}} P(s'|s,a,\theta)\\
    & w(\mu', \mu^{sas'}) \hat{V}^{*}(s',\mu')  + C(s,a,\theta), \phi\big)\Big], \forall s \in \mathcal{S}, \mu \in \hat{\mathcal{M}}. 
\end{align*}
}
\end{algorithmic}
\end{algorithm}

To interpolate all $\mu^{sas'}$ that can be reached from some $\mu_i \in \hat{\mathcal{M}}$ in one step, we use some convex combination of points $\mu_i$ in $\hat{\mathcal{M}}$. Let $w(\mu_i,\mu^{sas'})$ be the weight $w_i$ associated with $\mu_i$ when interpolating $\mu^{sas'}$. We can use this interpolation weight to define an approximate transition probability for posterior as: 
\begin{align*}
    \tilde{P}(\mu'|s,a,\mu,\theta) = \sum_{s' \in \mathcal{S}} P(s'|s,a,\theta) w(\mu',\mu^{sas'}).
\end{align*}

A sanity check that $\tilde{P}(\mu'|s,a,\mu,\theta)$ is indeed a transition probability: $\sum_{\mu' \in \mathcal{\hat{M}}} \tilde{P}(\mu'|s,a,\mu,\theta) = 1$ and $\tilde{P}(\mu'|s,a,\mu,\theta) \geq 0$. We choose the convex combination that minimizes the weighted Euclidean norm of the difference between $\mu$ and each $\mu_i$ by solving the following linear program:
\begin{align}\label{Weight LP}
    & \min_{w} \sum_{i} w_i ||\mu_i - \mu^{sas'}||_2^2 \\
    & \text{ s.t. } \sum_{i} w_i \mu_i(\theta) = \mu^{sas'}(\theta), \forall \theta \in \Theta \nonumber\\
    & ~~~~~~~~~~~~~~~~ \sum_i w_i = 1, w_i \geq 0, \forall i. \nonumber
\end{align}
With the approximation in the constraint in \eqref{DCP}, we obtain the following approximate bilevel DCP Algorithm~\ref{alg: approximate DCP} for a given posterior set. For ease of notation, we denote by $C(s,a,\theta)=\mathbb{E}_{\theta} [C(s,a,\xi)]$ the average cost at state $s$ when action $a$ is taken, under the parameter value $\theta$. 


\begin{theorem}\label{thm: lower bound}
The approximate value function $\hat{V}^{*}$ found by running Algorithm~\ref{alg: approximate DCP} is a lower bound on the exact optimal value function $V^{*}$. 
\end{theorem}

We also develop an upper bound on the exact optimal value function, using the obtained policy from Algorithm \ref{alg: approximate DCP}. The obtained policy is a finite state controller (see \citet{hansen2013solving} for the definition of finite state controller). Let $\mathcal{N}$ be the set of nodes in the controller such that we associate a node $n_{s,\mu}$ to each $(s,\mu)$ pair. The action chosen in node $n_{s,\mu}$ is determined by the policy $\hat{\pi}^{*}(a|s,\mu)$. For a given parameter $\theta$, the transition probability to the next node is $P(n_{s',\mu'}|n_{s,\mu},a)=w(\mu',\mu^{sas'}) P(s'|s,a)$. The value function of the finite state controller can be computed by
\begin{align*}
    \hat{V}^{\hat{\pi}^{*}}(n_{s,\mu}) & = \min_{\phi} \sum_{\theta \in \Theta} \mu(\theta)\Big[\Psi\big(c(s,a,\theta) + \gamma \sum_{n_{s',\mu'} \in \mathcal{N}} \\
    & ~~~~~w(\mu',\mu^{sas'}) P(s'|s,a,\theta) \hat{V}^{\hat{\pi}^{*}}(n_{s',\mu'}), \phi\big)\Big].
\end{align*}

Similar to \citet{ahmadi2021constrained}, the value function can be solved efficiently by DCP. It is also known from \citet{hansen2013solving} that the value function obtained by the finite state controller $\hat{V}^{\hat{\pi}^{*}}$ serves as an upper bound for the optimal value function. 

Note that the inequality $\hat{V}^{*} \leq V^{*} \leq \hat{V}^{\hat{\pi}^{*}}$ provides information about how well the optimal value function $V^{*}$ is approximated. As the posterior set $\hat{\mathcal{M}}$ contains more beliefs to accurately evaluate the policies, the gap between the approximate value function and the optimal value function gets smaller. 

\begin{algorithm}[ht]
\caption{New Posterior Set Generation}
\label{alg: belief generation}
\begin{algorithmic}
\STATE {\bfseries input:} policy $\hat{\pi}^{*}$, posterior set $\hat{\mathcal{M}}$, maximum number of newly added posterior distributions $n$
\STATE {\bfseries output:} newly added posterior set $\hat{\mathcal{M}}'$
\STATE initialization: $\hat{\mathcal{M}}' \xleftarrow{} \emptyset$.
\FOR{each $(s,\mu) \in (\mathcal{S} \times \hat{\mathcal{M}})$ and $s' \in \mathcal{S}$}
\STATE $\mu'(\theta) \propto \mu(\theta) f(\xi|\theta)$, where $s'=g(s,\hat{\pi}^{*}(a|s,\mu),\xi)$; 
\STATE $\operatorname{dist}_{\mu'} \xleftarrow{}$ distance of $\mu'$ to $\hat{\mathcal{M}} \bigcup \hat{\mathcal{M}}'$.
\IF{$\operatorname{dist}_{\mu'} > 0$ (i.e., $\mu'$ not in $\hat{\mathcal{M}} \bigcup \hat{\mathcal{M}}'$)}
\STATE $\hat{\mathcal{M}}' \xleftarrow{} \hat{\mathcal{M}}' \bigcup \{\mu'\}$.
\ENDIF
\IF{$|\hat{\mathcal{M}}'|>n$ (to reduce the size of $\hat{\mathcal{M}}'$)}
\FOR{each $\mu' \in \hat{\mathcal{M}}'$}
\STATE $\operatorname{dist}_{\mu'} \xleftarrow{}$ distance of $\mu'$ to $\hat{\mathcal{M}} \bigcup \hat{\mathcal{M}}' \backslash \{\mu'\}$; 
\STATE $\hat{\mathcal{M}}' \xleftarrow{} \hat{\mathcal{M}}' \backslash \{\argmin_{\mu' \in \hat{\mathcal{M}}'} \operatorname{dist}_{\mu'}\}$.
\ENDFOR
\ENDIF
\ENDFOR
\end{algorithmic}
\end{algorithm}

Next we incrementally add new posterior distributions to the posterior set $\hat{\mathcal{M}}$. Different methods can be employed to produce new posterior distributions that are added to the set $\hat{\mathcal{M}}$ at each iteration. We take a similar approach as \citet{poupart2015approximate}, which is based on envelope techniques. It considers the posterior distributions that can be reached in one step from any posterior distribution in $\hat{\mathcal{M}}$ by executing the policy $\hat{\pi}^{*}$. As the number of posterior distributions to be added might be excessive, we can prioritize them by including the $n$ reachable posterior distributions with the largest Euclidean distance to the posterior distributions in $\hat{\mathcal{M}}$. Note that the point-based value iteration approach in \citet{pineau2003point} shares the similar idea, that is, to include new posterior distribution that improves the worst-case density as rapidly as possible, where density is defined as the maximum distance from any posterior distribution to $\hat{\mathcal{M}}$. We summarize the new posterior set generation in Algorithm~\ref{alg: belief generation}.

Combining Algorithm \ref{alg: approximate DCP} and Algorithm \ref{alg: belief generation}, we now present the full algorithm below (ABDCP), which iteratively adds to the new posterior set and solves a bilevel difference convex program at each iteration. Theorem~\ref{thm: convergence} shows that Algorithm \ref{alg: full algorithm} converges to a near-optimal policy.

\begin{algorithm}[ht]
\caption{ABDCP for infinite-horizon BR-MDPs}
\label{alg: full algorithm}
\begin{algorithmic}
\STATE {\bfseries input:} threshold $\epsilon$, number of newly added posterior distributions $n$, initial state $s_1$, dataset $\mathcal{D}$
\STATE {\bfseries output:} policy $\hat{\pi}^{*}$
\STATE initialization: compute prior distribution $\mu_1$ using dataset $\mathcal{D}$; $\hat{M} \xleftarrow{} \{\text{degenerate beliefs}\} \bigcup \{\mu_1\}$.
\REPEAT
\STATE obtain $(\hat{\pi}^{*}, \hat{V}^{*})$ by running Algorithm \ref{alg: approximate DCP}; 
\STATE evaluate policy $\hat{\pi}^{*}$ by solving a DCP and obtain $\hat{V}^{\hat{\pi}^{*}}$;
\STATE $\hat{\mathcal{M}} \xleftarrow{} \hat{\mathcal{M}} \bigcup \hat{\mathcal{M}}'$ generated by Algorithm \ref{alg: belief generation}.
\UNTIL{$||\hat{V}^{\hat{\pi}^{*}} - \hat{V}^{*}||_{\infty} \leq \epsilon$}
\end{algorithmic}
\end{algorithm}

\begin{theorem}\label{thm: convergence}
Algorithm~\ref{alg: full algorithm} converges to a near-optimal policy $\hat{\pi}^{*}$, i.e., $||\hat{V}^{\hat{\pi}^{*}}- V^{*}||_{\infty} \leq \epsilon$, where $\epsilon$ is the desired threshold. 
\end{theorem}

As the number of iterations in Algorithm~\ref{alg: full algorithm} increases, the gap between the optimal value function and the lower bound becomes arbitrarily small, which shows the lower bound in Theorem~\ref{thm: lower bound} is non-trivial.

\section{Numerical Experiments}\label{sec: experiment}
We illustrate the performance of the infinite-horizon BR-MDP formulation with different choices of risk measures and the proposed approximate bilevel DCP algorithm with an offline path planning problem.

We adapt two methods to our offline planning problems and compare their performances. The first method (CALP) comes from \citet{poupart2015approximate} with a risk-neutral POMDP formulation. The second method (DR-MDP) comes from \citet{xu2010distributionally} with a distributionally robust MDP formulation. Note that the BPO approach from \citet{lee2018bayesian} solves a risk-neutral BA-MDP formulation, where two separate encoders for the physical state and belief state are designed to deal with the continuous latent parameter space. It could have been a good benchmark if its encoder design were made available. Apart from the two benchmarks, we also compare with the nominal approach (MLE), where a maximal likelihood estimator for the parameter is computed from the given dataset and then a policy is obtained by solving the MDP with the plugged-in parameter value. In our proposed algorithm (ABDCP) for the infinite-horizon BR-MDP formulation, we consider two particular risk measures, namely expectation and CVaR with different risk levels $\alpha$. It should be noted that, when the considered risk measure is expectation, our algorithm can be modified and reduced to CALP. Similar observation is verified in \citet{poupart2006analytic}, where the BA-MDP formulation is transformed into a POMDP formulation. 

\begin{table*}[ht]
\centering
\resizebox{14.2cm}{!}{%
\begin{tabular}{ccccc}
\hline
Approach                   & time (sec)    & expected cost        & CVaR ($\alpha=0.95$) cost & CVaR ($\alpha=0.8$) cost \\ \hline
ABDCP-EXP (CALP)           & 969.13(0.18)  & 70.06(0.51)          & 85.72                     & 82.06                    \\ \hline
ABDCP-CVaR ($\alpha=0.95$) & 2639.38(0.22) & 67.51(0.24)          & \textbf{75.67}                     & \textbf{73.72}          \\ \hline
ABDCP-CVaR ($\alpha=0.8$)  & 2545.74(0.24) & \textbf{66.02(0.38)}           & 79.97                     & 75.50                    \\ \hline
DR-MDP                     & 62.34(0.11)    & 79.43(0.15)          & 81.64                     & 80.60                    \\ \hline
Nominal                    & 61.44(0.08)    & 82.59(0.59)          & 94.10                     & 92.46                    \\ \hline
\end{tabular}
}
\caption{Results for path planning problem. Running time for each replication, expected cost, and CVaR cost at different risk levels $\alpha$ are reported for different algorithms. Standard errors are reported in parentheses. Number of data points is set to $N=10$. }
\label{table: path planning small}
\end{table*}

\begin{table*}[ht]
\centering
\resizebox{14.2cm}{!}{%
\begin{tabular}{ccccc}
\hline
Approach                   & time (sec)    & expected cost        & CVaR ($\alpha=0.95$) cost & CVaR ($\alpha=0.8$) cost \\ \hline
ABDCP-EXP (CALP)           & 967.25(0.17)  & 64.15(0.05)          & 66.34            & 65.97 \\ \hline
ABDCP-CVaR ($\alpha=0.95$) & 2642.26(0.21) & 65.18(0.03)          & 66.14            & 65.76 \\ \hline
ABDCP-CVaR ($\alpha=0.8$)  & 2643.48(0.25) & 65.17(0.04)          & 66.26            & 65.84 \\ \hline
DR-MDP                     & 63.15(0.09)    & 65.22(0.03)          & 66.43            & 66.01 \\ \hline
Nominal                    & 62.47(0.08)    & 64.31(0.12)          & 67.55            & 65.59 \\ \hline
\end{tabular}
}
\caption{Results for path planning problem. Running time for each replication, expected cost, and CVaR cost at different risk levels $\alpha$ are reported for different algorithms. Standard errors are reported in parentheses. Number of data points is set to $N=1000$. }
\label{table: path planning large}
\end{table*}

\begin{figure*}[!htp]
    \centering 
    \subfigure[ABDCP-EXP]{\label{figure: path planning expectation}\includegraphics[width=42mm]{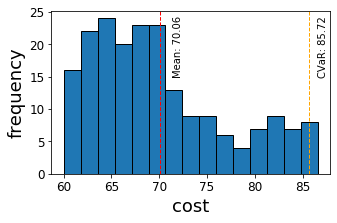}}
    \subfigure[ABDCP-CVaR($\alpha=0.95$)]{\label{figure: path planning CVaR 1}\includegraphics[width=42mm]{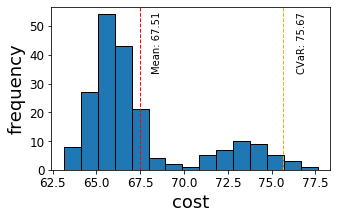}}
    \subfigure[ABDCP-CVaR($\alpha=0.8$)]{\label{figure: path planning CVaR 2}\includegraphics[width=42mm]{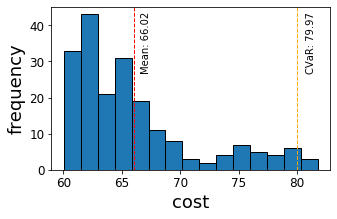}}   
    \subfigure[Nominal]{\label{figure: path planning MLE}\includegraphics[width=42mm]{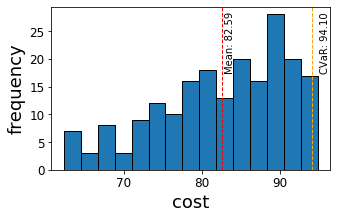}}   
    \captionsetup{justification=centering}
    \caption{Histogram of the actual performance over 200 replications for different algorithms. Number of data points $N=10$.}
    \label{figure: path planning histogram}
\end{figure*}  

For each of the considered algorithms, we obtain the corresponding optimal policy with the same dataset. It should be noted that the calculations are carried out offline. The obtained policy is then applied for risk-averse path planning and evaluated on the true system, i.e., MDP with the true parameter. This is referred to as one replication, and we repeat the experiments for 200 replications on different independent datasets. Results for the path planning problem can be found in Table \ref{table: path planning small} and Table \ref{table: path planning large}, with different data size $N=10$ and $N=1000$. The columns report the running time, expected performance (cost), and the CVaR performance (cost) of our proposed algorithm and benchmarks over the 200 replications. ABDCP-EXP stands for our proposed algorithm ABDCP with expectation as the risk measure. ABDCP-CVaR stands for our proposed algorithm ABDCP with CVaR as the risk measure. We also show the histogram of the actual performance over 200 replications for our proposed algorithm and the nominal benchmark on the path planning problem in Figure \ref{figure: path planning histogram}. We summarize the main observations for the path planning problem below.

\textbf{BR-MDP hedges against epistemic uncertainty}: in each replication, data points are randomly sampled from the true distribution. While facing the epistemic uncertainty, BR-MDP formulation optimizes over a dynamic risk measure that provides robustness. Table \ref{table: path planning small} shows that our proposed ABDCP algorithm is the most robust in the sense of balancing the mean and variability of the actual cost. The CVaR cost of our proposed algorithm is also lower than the other benchmarks, showing that it avoids large costs. In contrast, the nominal approach performs badly when the data size is small, e.g. $N=10$, indicating that it is not robust against the epistemic uncertainty and suffers from the scarcity of data. On the other hand, DR-MDP is overly conservative, even though it has the smallest variability. This conservativeness comes from two aspects. First, it always chooses to optimize over the worst-case scenario, which rarely happens in the true system. Second, the static worst-case risk measure prevents it from adapting to the data realizations, which is one of the motivations for the dynamic risk measure considered in the BR-MDP formulation. In contrast, BR-MDP formulation learns from the future data realization and updates its posterior distribution on $\theta$. 

\textbf{Larger data size reduces epistemic uncertainty}: when there are more data, the posterior distribution used in BR-MDP formulation and the MLE estimator used in the nominal approach converge to the true parameter, which reduces to solving an MDP with known transition probability and cost function. Therefore, the optimal policies and the actual costs tend to be the same. 

\textbf{Convergence of ABDCP}: the running time for a single replication on the path planning problem using our proposed ABDCP algorithm is affordable, and the proposed algorithm solves the infinite-horizon BR-MDP in finite time. In contrast, the infinite-horizon BR-MDP is intractable with standard value iteration or policy iteration.

\textbf{Effect of risk measures}: although both risk measures (expectation and CVaR) result in time-consistent optimal policy for each considered formulation, they provide different levels of robustness. Even though the expectation case is faster to compute, it provides the least robustness, especially when the data size is small. For the CVaR risk measure, different risk level $\alpha$ also affects the robustness. As $\alpha$ increases, the agent is more risk-averse, and the CVaR cost is smaller since it avoids more severe costs, as is shown in Figure \ref{figure: path planning CVaR 1} and Figure \ref{figure: path planning CVaR 2}. But this comes with a price: its expected cost is higher. This is intuitive: even though the agent avoids severe costs, it also forfeits a chance to traverse a path that is likely to have less traffic, even though the likelihood is small. This is shown as a right-shift of the actual performance distribution from Figure \ref{figure: path planning CVaR 2} to Figure \ref{figure: path planning CVaR 1}. 

\section{Conclusion}\label{sec: conclusions}
In this paper, we consider the offline planning problem in MDPs with epistemic uncertainty, where we only have access to a prior belief distribution over MDPs that is constructed by the offline data. We consider the infinite-horizon BR-MDP that produces a time-consistent formulation and provides the robustness against epistemic uncertainty. We develop an efficient optimization-based approximation algorithm that converges to the optimal policy. Our experimental results demonstrate the efficiency of the proposed approximate algorithm, and show the robustness of the infinite-horizon BR-MDP formulation. One of the future directions is to conduct the iteration complexity analysis on the proposed algorithm. Another interesting direction is to utilize function approximation to improve the scalability of the proposed approach to more complex domains. Separate encoders for the physical state and belief state have been proposed in \citet{lee2018bayesian} and adaptation from their risk-neutral BA-MDP formulation to our risk averse BR-MDP formulation could be interesting.

\section{Acknowledgments}
The authors gratefully acknowledge the support by the Air Force Office of Scientific Research under Grant FA9550-22-1-0244, the National Science Foundation under Grant NSF-ECCS-2419562 and the NSF AI Institute for Advances in Optimization under Grant NSF-2112533.

\bibliography{aaai25}

\clearpage
\appendix

\onecolumn

\section{Bayes Optimality and Offline Reinforcement Learning}
In Section~\ref{sec: Bellman} we define the optimal value function $V^{*}(s,\mu)=\min_{\pi \in \Pi^{MD}} V^{\pi}(s,\mu)$. Note that it is a function on the augmented state. This is consistent with the Bayesian reinforcement learning literature (e.g. \citet{duff2002optimal, rigter2021risk}), where the Bayes optimality is with respect to the prior knowledge of the system. In this paper, we consider the offline planning problem, where we only have access to the prior belief distribution over the MDP that is constructed with the offline data. Ideally, when the number of data points $N$ goes to infinity, one would expect the prior distribution already converges to a Dirac delta function concentrated on the true parameter $\theta^c$ almost surely (a.s.), and the optimal value function of BR-MDP, denoted by $V_N^{*}(s)$, converges to the optimal value function of the true MDP (denoted by $V^{*}(s)$). Note that $V_N^{*}(s)$ is equivalent to $V^{*}(s,\mu_N)$, where $\mu_N$ is the posterior distribution of $\theta$ after observing $N$ data points. 
\begin{assumption}\label{assumption:appendix}
(1) the cost function $C(\cdot,\cdot,\cdot)$ is bounded; (2) for every $\xi \in \Xi$, the distribution $f(\xi|\cdot)$ is continuous at $\theta^c$, and there is a measurable function $h: \Xi \to \mathbb{R}$ such that $\int_{\Xi} h(\xi) d\xi$ is finite and $f(\xi|\theta) \leq h(\xi)$ for all $\xi$ and all $\theta$ in a neighborhood of $\theta^c$; (3) For almost every sequence of data points of size $N$ and any $\epsilon>0$, $\lim_{N \to \infty} \int_{\theta: ||\theta-\theta^c||\leq \epsilon} \mu_N(\theta)d\theta=1$.
\end{assumption}

It should be noted that Assumption~\ref{assumption:appendix}(3) means that the Bayesian posterior of $\theta$ becomes more concentrated on $\theta^c$ as the data size increases and eventually degenerates to the Dirac delta function on $\theta^c$. It is a classical result under mild regularity conditions (see \citet{doob1949application}). 

\begin{theorem}
Under Assumption~\ref{assumption:appendix}, for a given dataset with size $N$, the solution $V_N^{*}$ to the infinite horizon BR-MDP converges uniformly to the solution $V^{*}$ to the true MDP (MDP with true parameter $\theta^c$), i.e., $||V_N^{*}-V^{*}||_{\infty} \to 0$, a.s., as $N \to \infty$. 
\end{theorem}

\begin{proof}
Following Definition~\ref{Bellman operator} in Section~\ref{sec: Bellman}, we define the following Bellman operator:
\begin{align*}
    \mathcal{T}^c V(s) &= \min_{a \in \mathcal{A}} \rho_{\delta(\theta^c)} \mathbb{E}_{\theta}[C(s,a,\xi)+\gamma V(s')]\\
    \mathcal{T}_N V(s) &= \min_{a \in \mathcal{A}} \rho_{\mu_N} \mathbb{E}_{\theta}[C(s,a,\xi)+\gamma V(s')].
\end{align*}
Note that as $N \to \infty$, the posterior distribution $\mu_N$ converges to a Dirac delta function concentrated on $\theta^c$, there is no need to update the posterior distribution as is in the BR-MDP formulation. By Lemma~\ref{Contraction mapping}, $\mathcal{T}^c$ and $\mathcal{T}_N$ are $\gamma$ contraction mapping for $||\cdot||_{\infty}$. Therefore, we have
\begin{align*}
    ||V_N^{*}-V^{*}||_{\infty} & = ||\mathcal{T}_N (V_N^{*}) -\mathcal{T}(V^{*})||_{\infty} \\
    & \leq ||\mathcal{T}_N (V^{*}_N) - \mathcal{T}_N (V^{*})||_{\infty} + ||\mathcal{T}_N (V^{*}) - \mathcal{T} (V^{*})||_{\infty} \\
    & \leq \gamma ||V_N^{*}-V^{*}|| + ||\mathcal{T}_N(V^{*}) - \mathcal{T}(V^{*})||_{\infty}.
\end{align*}
This gives
\begin{align*}
    ||V_N^{*}-V^{*}||_{\infty} \leq \frac{1}{1-\gamma} ||\mathcal{T}_N(V^{*}) - \mathcal{T}(V^{*})||_{\infty}. 
\end{align*}
Therefore, we need to show $\min_{a \in \mathcal{A}} \rho_{\mu_N} \mathbb{E}_{\theta}[C(s,a,\xi)+\gamma V^{*}(s')]$ converges uniformly to $\min_{a \in \mathcal{A}} \rho_{\delta(\theta^c)} \mathbb{E}_{\theta}[C(s,a,\xi)+\gamma V^{*}(s')]$.
For ease of notation, let $Z_{(s,a)}(\xi)=C(s,a,\xi)+\gamma V^{*}(s')$, and $Z_{(s,a)}:=Z_{(s,a)}(\theta)=\mathbb{E}_{\theta}[Z_{(s,a)}(\xi)]$. In turn, it suffices to show $\rho_{\mu_N}(Z_{(s,a)})$ converges to $\rho_{\delta(\theta^c)}(Z_{(s,a)})$ uniformly in $(s,a) \in \mathcal{S} \times \mathcal{A}$. Since $C(\cdot,\cdot,\cdot)$ and the value function are bounded, there is $C>0$ such that $|Z_{(s,a)}| \leq C$ for all $(s,a,\theta) \in \mathcal{S} \times \mathcal{A} \times \Theta$. For the risk measure considered in Section 2.3, $\rho_{\mu}(Z)=\inf_{\phi} \mathbb{E}_{\mu}[\Psi(Z,\phi)]$, we have that
\begin{align}\label{eq:difference}
    |\rho_{\mu_N}(Z_{(s,a)}) - \rho_{\delta(\theta^c)}(Z_{(s,a)})| = |\inf_{\phi} \int_{\Theta} \Psi(\theta, \phi) \mu_{N}(\theta)d\theta - \inf_{\phi} \int_{\Theta} \Psi(\theta, \phi) \delta(\theta^c)d\theta|.
\end{align}
Applying Theorem 5.3 in \cite{shapiro2021lectures}, as $N \to \infty$, $\mu_N(\theta) \to \delta(\theta^c)$ a.s., we have $\int_{\Theta} \Psi(\theta, \phi) \mu_{N}(\theta)d\theta \to \int_{\Theta} \Psi(\theta, \phi) \delta(\theta^c)d\theta$ a.s., and $\inf_{\phi} \int_{\Theta} \Psi(\theta, \phi) \mu_{N}(\theta)d\theta \to \inf_{\phi} \int_{\Theta} \Psi(\theta, \phi) \delta(\theta^c)d\theta$ a.s., which implies \eqref{eq:difference} $\to$ 0 a.s., thus the proof is complete.
\end{proof}

\section{ADDITIONAL LITERATURE REVIEW ON ROBUST MDPS}
There have been numerous approaches that address epistemic uncertainty in MDPs, with robust MDP being one of the most widely used methods. The robust MDP is typically formulated as a max-min problem. Here, the goal is to identify the policy that maximizes the value function while considering the most unfavorable model within an uncertainty set centered around a nominal model. There has been a growing commitment to improving the robust MDP framework over the recent years. 

For example, robust MDP framework usually assumes an $(s,a)$-rectangularity on the ambiguity set, which is too restrictive and may not hold in many applications. \citet{ho2022robust} study the robust MDP with a $\phi$-divergence ambiguity set under the $s$-rectangularity, which restricts the conservatism among transition probabilities corresponding to different actions in the same state and leads to a superior performance. \citet{goyal2023robust} adopt a factor model to represent probability transitions, wherein the transition probability is expressed as a linear function of a factor matrix that is uncertain and belongs to a factor matrix uncertainty set. This is more general and less restrictive than the $(s,a)$-rectangularity in most of the prior work, enabling the decision maker to capture interdependencies among probability transitions across different states. 

Also note that most of the prior work in robust MDP require the knowledge of the nominal model. \citet{panaganti2022sample} introduce a model-based RL algorithm aimed at acquiring an $\epsilon$-optimal robust policy in cases where the nominal model is unknown.

A final note to the robust MDP framework is that the prior work mostly focus on value iteration or policy iteration for solving the robust formulation. Recently, a set of policy gradient algorithms has emerged for solving the robust MDPs. \citet{wang2022policy} propose a policy gradient method for solving robust MDPs with a particular R-contamination ambiguity set. \citet{li2022first} propose a mirror descent approach for solving robust MDPs under the $(s,a)$-rectangularity. \citet{wang2023convergence} propose a double-loop policy gradient method for robust MDPs. They use two nested loops: an outer loop updates policies, and an inner loop approximately computes the worst-case transition probabilities over an infinite number of transition probabilities in the ambiguity set.

\section{Technical Proof}
\begin{proof}[Proof of Lemma \ref{Monotonicity}]
Note that
\begin{align*}
    (\mathcal{T}^{\pi} V)(s,\mu) = \rho_{\mu_1} \mathbb{E}_{\theta_1}[C(s_1,a_1.\xi_1) + \gamma \rho_{\mu_2} \mathbb{E}_{\theta_2}[C(s_2,a_2,\xi_2) + \cdots + \gamma V(s_k,\mu_k)]]
\end{align*}
for all positive integer $k$. As the terminal value function $V(s_k,\mu_k) \leq V'(s_k,\mu_k)$ for all $s_k \in \mathcal{S}, \mu_k \in \mathcal{M}$, and using the monotonicity of the convex risk measure $\rho(X) \leq \rho(Y)$ if $X(\omega) \leq Y(\omega), \forall \omega \in \Omega$, we have $(\mathcal{T}^{\pi} V)(s,\mu) \leq (\mathcal{T}^{\pi} V')(s,\mu)$. Same analysis works for operator $\mathcal{T}$. 
\end{proof}

\begin{proof}[Proof of Lemma \ref{Contraction mapping}]
Let $C_{\max} = \max_{s \in \mathcal{S}, \mu \in \mathcal{M}} |V(s,\mu)-V'(s,\mu)|$, we have 
\begin{align}\label{equation_1}
    V(s,\mu) - C_{\max} \leq V'(s,\mu) \leq V(s,\mu) + C_{\max}
\end{align}
Applying $\mathcal{T}^{\pi}$ on inequality \eqref{equation_1} and using the monotonicity shown in Lemma~\ref{Monotonicity}, we have
\begin{align*}
    (\mathcal{T}^{\pi} V)(s,\mu) - \gamma C_{\max} \leq (\mathcal{T}^{\pi} V')(s,\mu) \leq (\mathcal{T}^{\pi} V)(s,\mu) + \gamma C_{\max},
\end{align*}
which is justified by the translation invariance of the convex risk measure. Then we have
\begin{align*}
    \max_{s \in \mathcal{S}, \mu \in \mathcal{M}} |(\mathcal{T}^{\pi} V)(s,\mu) - (\mathcal{T}^{\pi} V')(s,\mu)| \leq \gamma \max_{s \in \mathcal{S}, \mu \in \mathcal{M}} |V(s,\mu)-V'(s,\mu)|,
\end{align*}
i.e., $||\mathcal{T}^{\pi} V - \mathcal{T}^{\pi} V'||_{\infty} \leq \gamma ||V-V'||_{\infty}$. Same analysis works for operator $\mathcal{T}$. 
\end{proof}

\begin{proof}[Proof of Proposition \ref{LP pre-condition}]
(i) Let $\pi$ be the policy for which 
\begin{align}\label{equation_2}
    V \geq \mathcal{T}^{\pi} V.
\end{align}
Note that such policy exists as one can choose $\pi$ that yields low current cost. Applying operator $\mathcal{T}^{\pi} V$ to both sides of inequality \eqref{equation_2} and using Lemma \ref{Monotonicity}, we have $V \geq (\mathcal{T}^{\pi})^{t} V$, $t=1,2,\cdots$. Note that the right hand side of the above inequality represents the cost of a finite horizon problem with stationary policy $\pi$ and with final value function $V$. Also note that 
\begin{align*}
    (\mathcal{T}^{\pi})^{t} V & = \rho_{\mu_1} \mathbb{E}_{\theta_1}[C(s_1,a_1,\xi_1) + \cdots + \gamma V(s_{t+1},\mu_{t+1})] \\
    & \geq \rho_{\mu_1} \mathbb{E}_{\theta_1}[C(s_1,a_1,\xi_1) + \cdots + \gamma \rho_{\mu_t}\mathbb{E}_{\theta_t}[C(s_t,a_t,\xi_t)]].
\end{align*}
Let $t \to \infty$, we get $V \geq V^{\pi}$, where $V^{\pi}$ is the value function under policy $\pi$. Since $V^{*} = \min_{\pi} V^{\pi}$, we have $V \geq V^{*}$.

(ii) Consider an arbitrary policy $\pi$ and a finite horizon problem with terminal cost $V(s_{t+1},\mu_{t+1})$. We have under the policy $\pi$,
\begin{align*}
    & \rho_{\mu_1} \mathbb{E}_{\theta_1} [C(s_1,a_1,\xi_1) + \cdots + \gamma V(s_{t+1},\mu_{t+1})] \\
    = & \rho_{\mu_1} \mathbb{E}_{\theta_1} [C(s_1,a_1,\xi_1) + \cdots + \gamma \rho_{\mu_t} \mathbb{E}_{\theta_t}[C(s_t,a_t,\xi_t) + \gamma V(s_{t+1},\mu_{t+1})]].
\end{align*}
Note that $\rho_{\mu_t} \mathbb{E}_{\theta_t} [C(s_t,a_t,\xi_t) + \gamma V(s_{t+1}, \mu_{t+1})] \geq \mathcal{T} V (s_t,\mu_t) \geq V(s_t,\mu_t)$. Therefore, we have 
\begin{align*}
    & \rho_{\mu_1} \mathbb{E}_{\theta_1} [C(s_1,a_1,\xi_1) + \cdots + \gamma V(s_{t+1}, \mu_{t+1})] \\
    \geq & \rho_{\mu_1} \mathbb{E}_{\theta_1} [C(s_1,a_1,\xi_1) + \cdots + \rho_{\mu_{t-1}} \mathbb{E}_{\theta_{t-1}}[C(s_{t-1},a_{t-1},\xi_{t-1}) + \gamma V(s_t,\mu_t)]].
\end{align*}
Continuing, we have
\begin{align*}
    \rho_{\mu_1} \mathbb{E}_{\theta_1} [C(s_1,a_1,\xi_1) + \cdots + \gamma V(s_{t+1},\mu_{t+1})] \geq V(s,\mu).
\end{align*}
Let $C_{\max}$ be an upper bound on $|V(s,\mu)|, \forall s \in \mathcal{S}, \mu \in \mathcal{M}$, we have 
\begin{align*}
    \rho_{\mu_1} \mathbb{E}_{\theta_1} [C(s_1,a_1,\xi_1) + \cdots + \gamma \rho_{\mu_t} \mathbb{E}_{\theta_t}[C(s_t,a_t,\xi_t)]] \geq V(s,\mu) - C_{\max} \gamma^{t}.
\end{align*}
Passing to the limit $t \to \infty$, we have for any policy $\pi$, $V^{\pi}(s,\mu) \geq V(s,\mu)$. Therefore, the infimum over all policy $\pi$ is bounded from below by $V(s,\mu)$, i.e., $V^{*} \geq V$.
\end{proof}

\begin{proof}[Proof of Theorem \ref{thm: lower bound}]
We first show $V^{\pi}(s,\mu)$ is concave in $\mu$ for any policy $\pi$. Note that 
\begin{align*}
    V^{\pi}(s,\mu) = \min_{\phi} \mathbb{E}_{\mu}[\Psi(\mathbb{E}_{\theta}[C(s,a,\xi) + \gamma V^{\pi}(s',\mu')], \phi)].
\end{align*}
For $0 \leq t \leq 1$ and any $\mu_1, \mu_2 \in \mathcal{M}$,
\begin{align*}
    & V^{\pi}(s, t \mu_1 + (1-t) \mu_2) \\
    = & \min_{\phi} \mathbb{E}_{t \mu_1 + (1-t) \mu_2} [\Psi(\mathbb{E}_{\theta} [C(s,a,\xi)+\gamma V^{\pi}(s',\mu')], \phi)] \\
    = & \min_{\phi} \{ \mathbb{E}_{\mu_1} [\Psi(\mathbb{E}_{\theta} [C(s,a,\xi)+\gamma V^{\pi}(s',\mu')], \phi)] \\
    & \quad + (1-t) \mathbb{E}_{\mu_2} [\Psi(\mathbb{E}_{\theta} [C(s,a,\xi)+\gamma V^{\pi}(s',\mu')], \phi)]\} \\
    \geq & t \min_{\phi_1}\{\mathbb{E}_{\mu_1} [\Psi(\mathbb{E}_{\theta} [C(s,a,\xi)+\gamma V^{\pi}(s',\mu')], \phi_1)]\} \\
    & \quad + (1-t) \min_{\phi_2}\{\mathbb{E}_{\mu_2} [\Psi(\mathbb{E}_{\theta} [C(s,a,\xi)+\gamma V^{\pi}(s',\mu')], \phi_2)]\} \\
    = & t V^{\pi}(s,\mu_1) + (1-t) V^{\pi}(s,\mu_2).
\end{align*}
The same analysis works for the optimal value function $V^{*}$. Consider running Algorithm~\ref{alg: approximate DCP} with posterior set $\hat{\mathcal{M}}$ and the entire posterior set $\mathcal{M}$. Now applying Jensen's inequality and by Theorem 12 in \citet{hauskrecht2000value}, we have $\hat{V}^{*} \leq V^{*}$. Note that originally in \citet{hauskrecht2000value}, the proof is based on the fact that value function in partially observable Markov decision process is convex in belief and the linear programming formulation has constraint $V(b) \geq R(s,a) + \gamma \sum_{b'}P(b'|b,a)V(b')$, where $R$ is the reward function. Since $V$ is convex and by linear interpolation, applying Jensen's inequality to the right hand side of the constraint leads to $\hat{V}(b)$ greater than $V(b)$. Now we are in an opposite direction, by Jensen's inequality and concavity of $V$, we have $\hat{V}^{*} \leq V^{*}$.
\end{proof}

\begin{proof}[Proof of Theorem \ref{thm: convergence}]
First, we show Algorithm~\ref{alg: belief generation} is essentially a greedy algorithm for geometric set cover problem (see \citet{clarkson2005improved}). Denote by $K$ the cardinality of the parameter space $\Theta$. The set $\mathcal{M}$ of posterior distribution $\mu$ is then a $K-1$ dimensional simplex (probability simplex). We want to find a subset from a set of $\epsilon$-size ball of dimension $K-1$ so as to cover all points in the probability simplex $\mathcal{M}$ in finite time, which is categorized as geometric set cover problem (except that we do not require to find the subset of minimum size). For example, when $K=2$, the problem is reduced to finding a union of intervals with length $\epsilon$ to cover the line segment $\{(x,y)|x+y=1,0\leq x \leq 1, 0 \leq y \leq 1\}$; when $K=3$, the problem is reduced to finding a union of circles of radius $\epsilon$ to cover the plane $\{(x,y,z)|x+y+z=1,0\leq x \leq 1, 0 \leq y \leq 1,0\leq z \leq 1\}$. 

Suppose we are at iteration $k$ in Algorithm~\ref{alg: full algorithm}. Algorithm~\ref{alg: belief generation} adds to the current posterior set $\hat{\mathcal{M}}_k$ reachable posterior distributions that are farthest from the current set $\hat{\mathcal{M}}_k$, which ensures the number of uncovered posterior distributions added to the current set $\hat{\mathcal{M}}_k$ is maximized. For set cover, we evaluate an approximation algorithm by considering the ratio between the
number of subsets used in the cover output by the algorithm and the number of subsets used by the optimal solution. The greedy algorithm gives an approximation ratio that is very close to the best for a polynomial-time algorithm (see \citet{feige1998threshold, clarkson2005improved}).

According to Algorithm~\ref{alg: belief generation}, after a finite time $k$, the current posterior set $\hat{\mathcal{M}}_k$ contains enough posterior distributions such that $\forall s \in \mathcal{S}, a \in \mathcal{A}, s' \in \mathcal{S}, \mu \in \hat{\mathcal{M}}_k$, $\exists \mu' \in \hat{\mathcal{M}}_k$, s.t. $||\mu^{sas'}-\mu'||_2 \leq \epsilon$, $\forall \epsilon' > 0$. Recall that the Bellman operator $\mathcal{T}$ in Definition~\ref{Bellman operator} can be written as:
\begin{align*}
    \mathcal{T} V(s,\mu)=\min_{\phi} \sum_{\theta \in \Theta}\mu(\theta)\Big[\Psi\big(C(s,a,\theta)+\gamma \underbrace{\sum_{s' \in \mathcal{S}, \mu' \in \mathcal{M}}P(s'|s,a,\theta)V(s',\mu')}_{Z}, \phi\big)\Big].
\end{align*}
Denote by $\mathcal{T}_k$ the Bellman operator in the $k$-th iteration of the Algorithm~\ref{alg: full algorithm}, i.e.,
\begin{align*}
    \mathcal{T}_k V(s,\mu)=\min_{\phi} \sum_{\theta \in \Theta}\mu(\theta)\Big[\Psi\big(C(s,a,\theta)+\gamma \underbrace{\sum_{s' \in \mathcal{S}, \mu' \in \hat{\mathcal{M}}_k}P(s'|s,a,\theta)w(\mu',\mu^{sas'})V(s',\mu')}_{Z_k}, \phi\big)\Big].
\end{align*}
Considering the set $\{s'\in \mathcal{S}|\mu^{sas'}\in\hat{\mathcal{M}}_k\}$ and $\{s'\in \mathcal{S}|\mu^{sas'}\notin\hat{\mathcal{M}}_k\}$, we can rewrite $\mathcal{T}_k$ as:
\begin{align*}
    \mathcal{T}_k V(s,\mu)=\min_{\phi} \sum_{\theta \in \Theta} \mu(\theta)\Big[&\Psi\big(C(s,a,\theta)+\gamma \sum_{s' \in \mathcal{S}, \mu^{sas'}\in\hat{\mathcal{M}}_k}P(s'|s,a,\theta)V(s',\mu^{sas'})\\
    & +\gamma \sum_{s' \in \mathcal{S}, \mu^{sas'}\notin\hat{\mathcal{M}}_k,\mu'\in\hat{\mathcal{M}}_k}P(s'|s,a,\theta)w(\mu',\mu^{sas'})V(s',\mu'),\phi\big)\Big].
\end{align*}
We then bound the difference between $Z_k$ and $Z$ as follows.

{\small
\begin{align*}
    & ||Z-Z_k||_{\infty} \\
    = & ||\sum_{s' \in \mathcal{S}, \mu^{sas'} \in \mathcal{M}} P(s'|s,a,\theta)V(s',\mu^{sas'}) - \sum_{s' \in \mathcal{S}, \mu^{sas'}\in\hat{\mathcal{M}}_k}P(s'|s,a,\theta)V(s',\mu^{sas'}) \\
    & -  \sum_{s' \in \mathcal{S}, \mu^{sas'}\notin\hat{\mathcal{M}}_k,\mu'\in\hat{\mathcal{M}}_k}P(s'|s,a,\theta)w(\mu',\mu^{sas'})V(s',\mu')||_{\infty}\\
    = & ||\sum_{s' \in \mathcal{S}, \mu^{sas'} \notin \hat{\mathcal{M}}_k} P(s'|s,a,\theta)V(s',\mu^{sas'}) - \sum_{s' \in \mathcal{S}, \mu^{sas'}\notin\hat{\mathcal{M}}_k,\mu'\in\hat{\mathcal{M}}_k}P(s'|s,a,\theta)w(\mu',\mu^{sas'})V(s',\mu')||_{\infty} \\
    \leq & ||\sum_{s' \in \mathcal{S}, \mu^{sas'} \notin \hat{\mathcal{M}}_k} P(s'|s,a,\theta)V(s',\mu^{sas'}) - \big(\sum_{s' \in \mathcal{S}, \mu^{sas'} \notin \hat{\mathcal{M}}_k, ||\mu^{sas'}-\mu'|| \leq \epsilon'}P(s'|s,a,\theta)(1-\delta\epsilon')V(s',\mu')\\
    & + \sum_{s' \in \mathcal{S}, \mu^{sas'} \notin \hat{\mathcal{M}}_k, ||\mu^{sas'}-\mu'|| > \epsilon'}P(s'|s,a,\theta)\delta\epsilon' V(s',\mu')\big)||_{\infty}\\
    \leq & ||\sum_{s' \in \mathcal{S}, \mu^{sas'} \notin \hat{\mathcal{M}}_k} P(s'|s,a,\theta)V(s', \mu^{sas'}) - \sum_{s' \in \mathcal{S}, \mu^{sas'} \notin \hat{\mathcal{M}}_k, ||\mu^{sas'}-\mu'|| \leq \epsilon'}P(s'|s,a,\theta)V(s',\mu')||_{\infty}\\
    & + \delta \epsilon' ||\sum_{s' \in \mathcal{S}, \mu^{sas'} \notin \hat{\mathcal{M}}_k, ||\mu^{sas'}-\mu'|| \leq \epsilon'}P(s'|s,a,\theta) V(s',\mu')-\sum_{s' \in \mathcal{S}, \mu^{sas'} \notin \hat{\mathcal{M}}_k, ||\mu^{sas'}-\mu'|| > \epsilon'}P(s'|s,a,\theta) V(s',\mu')||_{\infty} \\
    \leq & \epsilon' L_{\mu} + \delta \Delta \epsilon':=A\epsilon'.
\end{align*}
}The first inequality holds because the optimal weight in (5) satisfies: $\exists \delta > 0$ s.t. $w(\mu',\mu^{sas'}) \geq 1-\delta \epsilon'$, where $||\mu'-\mu^{sas'}||\leq \epsilon'$ for some constant $\delta$ that depends on $K$, which is the cardinality of the parameter space $\Theta$. It essentially implies that the closest belief point $\mu'$ that is within the $\epsilon'$-distance will always be included in the convex combination of $\mu^{sas'}$, and the weight is close to 1 as to minimize the weighted euclidean norm. The second inequality holds because of the triangle inequality. In Theorem~\ref{thm: lower bound} we have shown $V^{\pi}(s,\mu)$ is concave in $\mu$ for any policy $\pi$. Since concavity implies Lipschitz continuity, we have $||V(s',\mu^{sas'})-V(s',\mu')||_{\infty}\leq L_{\mu}||\mu^{sas'}-\mu'||\leq L_{\mu} \epsilon'$, where $L_{\mu}$ is the Lipschitz constant. Since we assume the cost is bounded, which implies the value function is also bounded, we have $||\sum_{s' \in \mathcal{S}, \mu^{sas'} \notin \hat{\mathcal{M}}_k, ||\mu^{sas'}-\mu'|| \leq \epsilon'}P(s'|s,a,\theta) V(s',\mu')-\sum_{s' \in \mathcal{S}, \mu^{sas'} \notin \hat{\mathcal{M}}_k, ||\mu^{sas'}-\mu'|| > \epsilon'}P(s'|s,a,\theta) V(s',\mu')||_{\infty}\leq\Delta$ for some constant $\Delta>0$. So the third inequality holds. Next, note that $\Psi(Z,\phi)$ is convex in $(z,\phi)$, it remains to be convex in $Z$ after taking the minimum over $\phi$. Since convexity implies Lipschitz continuity (with Lipschitz constant $L_z$), we have 
\begin{align*}
    ||\mathcal{T}_k V - \mathcal{T} V||_{\infty} \leq A \gamma L_{z} \epsilon'.
\end{align*}
Last, denote by $\hat{\pi}_k$ the policy obtained in the $k$-th iteration in Algorithm~\ref{alg: full algorithm}, we have 
\begin{align*}
    ||\hat{V}^{\hat{\pi}_k}-V^{*}||_{\infty}&=||\mathcal{T}_k \hat{V}^{\hat{\pi}_k} - \mathcal{T} V^{*} ||_{\infty} \\
    & = ||\mathcal{T}_k \hat{V}^{\hat{\pi}_k} - \mathcal{T} \hat{V}^{\hat{\pi}_k} + \mathcal{T} \hat{V}^{\hat{\pi}_k} - \mathcal{T} V^{*}||_{\infty} \\
    & \leq ||\mathcal{T}_k \hat{V}^{\hat{\pi}_k} - \mathcal{T} \hat{V}^{\hat{\pi}_k}||_{\infty} + ||\mathcal{T} \hat{V}^{\hat{\pi}_k} - \mathcal{T} V^{*}||_{\infty}\\
    & \leq A \gamma L_{z} \epsilon' + \gamma || \hat{V}^{\hat{\pi}_k} - V^{*}||_{\infty}.
\end{align*}
Therefore, we have $||\hat{V}^{\hat{\pi}_k}-V^{*}||_{\infty} \leq \frac{A \gamma L_{z} \epsilon'}{1-\gamma}$. By setting $\epsilon'=\frac{1-\gamma}{A \gamma L_{z}}\epsilon$, the proof is complete.
\end{proof}

\section{Implementation Details}
All algorithms are implemented in Python and run on a 1.4 GHz Intel Core i5 processor with 8 GB memory. 
\subsection{Offline Path Planning}
\begin{figure}[ht]
    \centering
    \includegraphics[width=0.4\textwidth]{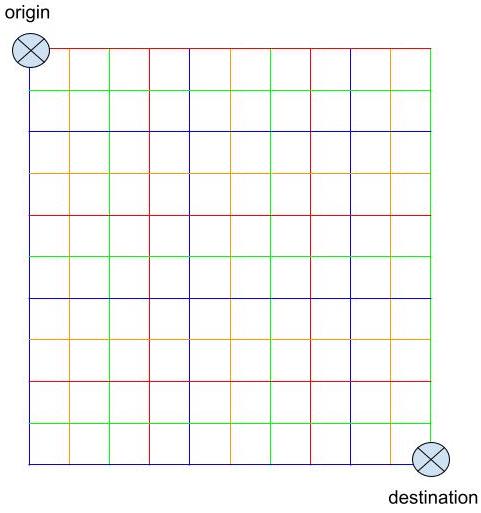}
    \caption{Path planning terrain map. Colors indicate the road types as follows--blue: highway, red: main road, orange: street, green: lane.}
    \label{figure: path planning map}
\end{figure} 

An autonomous car (agent) navigates a two-dimensional terrain map represented by a 10 by 10 grid along roads to the destination, as shown in Figure \ref{figure: path planning map}. The agent chooses from four actions \{up, down, left, right\}, as long it remains on the road. There are four types of roads: \{highway, main road, street, lane\}. The traffic time $\xi^{T}_i$ in each type of road is assumed to be independent and follows exponential distribution with different rate, denoted by $\theta_i^{T}, i=1,\cdots,4$, where $T$ stands for traffic time. Specifically, the true rates are $\theta_1^{T}=1$, $\theta_2^{T}=0.5$, $\theta_3^{T}=0.2$, and $\theta_4^{T}=0.1$ but unknown to the agent. We view the parameter as a random variable, whose value is assumed to be within the following finite set $\{0.05,0.1,0.2,0.25,0.3,0.4,0.5,0.6,0.7,0.8,0.9,1,1.5,2,2.5\}$. $\xi_i^{A} \in \{0,1\}, i=1,\cdots,4$
denotes whether there is car accident in each type of road, where $A$ stands for accident. The probability of car accident happening in each type of road is also assumed to be independent, denoted by $\theta_i^{A}$. Specifically, the true probabilities are $\theta_1^{A}=0.3$, $\theta_2^{A}=0.2$, $\theta_3^{A}=0.1$ and $\theta_4^{A}=0.05$ but unknown to the agent. We view the parameter as a random variable, whose value is assumed to be within the following finite set $\{0.05,0.1,0.15,0.2,0.25,0.3,0.35,0.4,0.45,0.5\}$. When there is an car accident, the agent receives a constant cost $T_A=10$ and makes no transition. Otherwise, the agent transitions to the next road depending on the action it takes and receives the cost, which is the traffic time for traversing that type of road. The agent stops when it reaches the destination. The discount factor $\gamma = 0.95$. The agent is given a historical dataset $\mathcal{H}_0$ of size $N$ containing past traffic times and car accident logs, and uses the given dataset to construct the prior for the transition rate and probability of car accident. Other parameters are as follows: number of points to be added at each iteration $n=20$, threshold $\epsilon=0.1$.  

\subsection{Multi-item Inventory Control}
The warehouse manager (agent) decides how much to replenish from the set $\{0,1,\cdots,S_i-s_i\}$ for each item $i \in [K]$ at each time stage, where $K=5$ is the number of different items, $S_i=100$ is the storage capacity for each item $i$, $s_i$ is the current inventory level for each item $i$. The customer demand is a random vector $\xi=(\xi^1,\cdots,\xi^K)$ with each $\xi_i$ following a Poisson distribution with parameter $\theta^i$. The true parameter is $\theta^1 = 10$, $\theta^2 = 15$, $\theta^3 = 20$, $\theta^4 = 25$, and $\theta^5 = 30$ but unknown to the agent. We view the parameter as a random variable, whose value is assumed to be within the following finite set $\{5, 6, 7, \cdots, 33, 34, 35\}$. The state transition is given by $s_{t+1}=\max(s_t+a_t-\xi_t, 0)$, where $a_t$ is the amount of inventory to be replenished. Inventory level is not allowed to drop below zero (no backlog). When the customer demand is higher than the supply, there is a penalty cost $p$ for each unit of unsatisfied demand. When the customer is lower than the supply, there is a holding cost $h$ for each unit of overstock. In particular, for different items, $p_1=4$, $p_2=5$, $p_3=6$, $p_4=7$, $p_5=8$, $h_1=2$, $h_2=3$, $h_3=4$, $h_4=5$, $h_5=6$. The cost function at each stage is then given by $C(s_t,a_t,\xi_t)=h^{T} \cdot \max(s_t + a_t - \xi_t, 0) + p^{T} \cdot \max(\xi_t - s_t - a_t, 0)$. The discount factor $\gamma = 0.95$. The agent starts with 0 inventory and is given a historical dataset $\mathcal{H}_0$ of size $N$ containing past customer demands for different items, and uses the given dataset to construct the prior for the rate parameter. Other parameters are as follows: number of points to be added at each iteration $n=20$, threshold $\epsilon=0.1$.  

\subsection{DR-MDP Details}
The DR-MDP approach, or Distributionally Robust Markov Decision Process, is a method for decision making under uncertainty where the ambiguity set, or the set of possible distributions for the uncertain parameters, is constructed using prior knowledge about the probabilistic information. However, this prior knowledge is not always readily available from a given data set, making the construction of the ambiguity set difficult in some cases.

We note that the Bayesian Risk Optimization (BRO) approach has a distributionally robust optimization (DRO) interpretation. In particular, for a static stochastic optimization problem, it has been shown in \citet{wu2018bayesian} that the BRO formulation with the risk functional taken as Value-at-Risk (VaR) with a confidence level of $100\%$ is equivalent to a DRO formulation with the ambiguity set constructed for the uncertain parameter, $\theta$. This means that BRO and DRO can be used interchangeably, depending on the problem at hand and the level of uncertainty and prior knowledge about the parameters. Therefore, for a given problem when prior knowledge about the probabilistic information is not readily available, we adapt DR-MDP to our considered problem as follows: we use samples of the uncertain parameter, $\theta$, drawn from the posterior distribution computed from a given data set. This allows us to construct an ambiguity set for $\theta$ using the available data, instead of relying on prior knowledge. Once we have samples of $\theta$, we can obtain the optimal policy that minimizes the total expected cost under the most adversarial $\theta$ among the samples.

\subsection{Bilevel Optimization}
For the considered class of convex risk measures, we can rewrite the above formulation as a bilevel difference convex program:
\begin{align}
    & \min_V \quad -\sum_{s \in \mathcal{S},\mu \in \mathcal{M}} \alpha(s,\mu) V(s,\mu) \label{DCP_2} \\
    & \text{s.t.} V(s,\mu) - \min_{\phi} \mathbb{E}_{\mu}\Big[\Psi\big(\mathbb{E}_{\theta}[C(s,a,\xi) + \gamma V(s',\mu')], \phi\big)\Big] \leq 0, \forall a \in \mathcal{A}, s \in \mathcal{S}, \mu \in \mathcal{M} \nonumber.
\end{align}

We show the bilevel DCP can be reduced to a single-level DCP. Specifically, we show this transformation for the exact bilevel DCP in \eqref{DCP_2}, and the same technique can be applied to the approximate algorithm. 

Consider a general bilevel optimization problem:
\begin{align}\label{general bilevel optimization}
    & \min_{x_u, x_l} F(x_u, x_l) \\
    & \text{s.t. } x_l \in \argmin_{x_l} \{f(x_u,x_l): g(x_u,x_l) \leq 0\} \nonumber \\
    & ~~~~~~ G(x_u,x_l) \leq 0 \nonumber
\end{align}
where $x_u$ is the upper-level variable, $x_l$ is the lower-level variable, $G$ denotes the upper-level constraints, $g$ denotes the lower-level constraints, $F$ denotes the upper-level objective function, $f$ denotes the lower-level objective function. The Karush-Kuhn-Tucker (KKT) conditions are a set of necessary and sufficient conditions for a solution to be optimal in a convex optimization problem. When the lower-level problem in a bilevel optimization problem is convex and sufficiently regular, the KKT conditions can be used to reformulate the problem as a single-level constrained optimization problem, which is typically easier to solve. The general bilevel optimization problem \ref{general bilevel optimization} can then be reduced to the following single-level optimization:
\begin{align*}
    & \min_{x_u, x_l} F(x_u, x_l) \\
    & \text{s.t. } G(x_u,x_l) \leq 0 \\
    & ~~~~~~ \nabla_{x_l} L(x_u,x_l,\lambda)=0 \\
    & ~~~~~~ g(x_u,x_l) \leq 0 \\
    & ~~~~~~ \lambda g(x_u,x_l)=0 \\
    & ~~~~~~ \lambda \geq 0
\end{align*}
where $L(x_u,x_l,\lambda)=f(x_u,x_l)+\lambda g(x_u,x_l)$ is the Lagrangian function. In the bilevel DCP \eqref{DCP_2}, $V$ is the upper-level variable and $\phi$ is the lower-level variable. The constraint in \eqref{DCP_2} can be rewritten as: 
\begin{align*}
    & \phi \in \argmin_{\phi \in \Phi} \mathbb{E}_{\mu} [\Psi(\mathbb{E}_{\theta}[C(s,a,\xi)+\gamma V(s',\mu')], \phi)] \\
    & V(s,\mu) - \mathbb{E}_{\mu} [\Psi(\mathbb{E}_{\theta}[C(s,a,\xi)+\gamma V(s',\mu')], \phi)] \leq 0.
\end{align*}

Since the lower-level problem is convex, we can reformulate the bilevel DCP \eqref{DCP_2} as a single-level DCP problem.
\begin{align*}
    & \min_V - \sum_{s \in \mathcal{S},\mu \in \mathcal{M}} \alpha(s,\mu) V(s,\mu) \\
    & \text{s.t.} V(s,\mu) - \mathbb{E}_{\mu}[\Psi(\mathbb{E}_{\theta}[C(s,a,\xi) + \gamma V(s',\mu')], \phi)] \leq 0 \\
    & ~~~~~ \nabla_{\phi} \mathbb{E}_{\mu}[\Psi(\mathbb{E}_{\theta}[C(s,a,\xi) + \gamma V(s',\mu')], \phi)] = 0 \\
    & ~~~~~~ \forall a \in \mathcal{A}, s \in \mathcal{S}, \mu \in \mathcal{M}.
\end{align*}

\clearpage

\subsection{Additional Experiments}
Results for the multi-item inventory control problem are reported in Table \ref{table: inventory control small} and Table \ref{table: inventory control large}. 

\begin{table}[ht]
\centering
\resizebox{14.2cm}{!}{%
\begin{tabular}{ccccc}
\hline
Approach                   & time (sec)    & expected cost  & CVaR ($\alpha=0.95$) cost & CVaR ($\alpha=0.8$) cost \\ \hline
ABDCP-EXP (CALP)           & 1374.59(0.24) & 3478.92(15.03) & 4363.56                   & 4025.23                  \\ \hline
ABDCP-CVaR ($\alpha=0.95$) & 4109.21(0.35) & 3072.84(10.24) & 3651.75          & 3472.44         \\ \hline
ABDCP-CVaR ($\alpha=0.8$)  & 4087.14(0.32) & 2831.12(12.37) & 3782.04          & 3517.30                  \\ \hline
DR-MDP                     & 140.76(0.13)  & 3963.56(9.21)  & 4424.69                   & 4231.12                  \\ \hline
Nominal                    & 139.89(0.08)  & 3987.50(18.39) & 4974.57                   & 4611.16                  \\ \hline
\end{tabular}
}
\caption{Results for multi-item inventory control problem. Running time for each replication, expected cost, and CVaR cost at different risk levels $\alpha$ are reported for different algorithms. Standard errors are reported in parentheses. Number of data points is $N=10$. }
\label{table: inventory control small}
\end{table}

\begin{table}[!ht]
\centering
\resizebox{14.2cm}{!}{%
\begin{tabular}{ccccc}
\hline
Approach                   & time (sec)    & expected cost        & CVaR ($\alpha=0.95$) cost & CVaR ($\alpha=0.8$) cost \\ \hline
ABDCP-EXP (CALP)           & 1377.27(0.22)  & 1806.45(0.57)          & 1825.06                     & 1823.71                    \\ \hline
ABDCP-CVaR ($\alpha=0.95$) & 4114.93(0.34) & 1819.92(0.16)          & 1823.63            & 1821.70           \\ \hline
ABDCP-CVaR ($\alpha=0.8$)  & 4002.62(0.33) & 1817.21(0.19) & 1824.97                     & 1822.39                    \\ \hline
DR-MDP                     & 138.26(0.11)    & 1826.03(0.12)          & 1828.80                     & 1827.62                   \\ \hline
Nominal                    & 136.38(0.10)    & 1802.34(1.28)        & 1836.04                   & 1828.99                 \\ \hline
\end{tabular}
}
\caption{Results for multi-item inventory control problem. Running time for each replication, expected cost, and CVaR cost at different risk levels $\alpha$ are reported for different algorithms. Standard errors are reported in parentheses. Number of data points is $N=1000$. }
\label{table: inventory control large}
\end{table}

\end{document}